\documentclass{article}

\usepackage[preprint]{neurips_2023}

\usepackage[utf8]{inputenc}
\usepackage[T1]{fontenc}
\usepackage{hyperref}
\usepackage{url}
\usepackage{booktabs}
\usepackage{amsmath}
\usepackage{amsfonts}
\usepackage{amssymb}
\usepackage{amsthm}
\usepackage{nicefrac}
\usepackage{microtype}
\usepackage{xcolor}
\usepackage{graphicx}

\newtheorem{assumption}{Assumption}
\newtheorem{lemma}{Lemma}
\newtheorem{proposition}{Proposition}
\newtheorem{theorem}{Theorem}
\DeclareMathOperator{\Var}{Var}

\title{Restricted Entropic and Quadratic Hedging in the Small-Risk-Aversion Limit}

\author{%
  Jonas~Hanetho \\
  Department of Mathematics\\
  University of Oslo\\
  \texttt{jonasrha@ifi.uio.no} \\
}

\begin{document}

\maketitle

\begin{abstract}
In incomplete markets, contingent claims cannot in general be replicated exactly, and hedging must therefore be formulated as an optimization problem over admissible trading strategies. This paper studies the relation between quadratic hedging and entropic hedging in a discrete-time incomplete market when strategies are restricted to bounded finite-dimensional parametric classes. These restrictions induce admissible gain sets that are compact and convex in \(L^2\), providing a common well-posed setting for both the restricted mean-square hedging problem and the restricted exponential-utility hedging problem. The main contribution is to establish a uniform small-risk-aversion expansion of the entropic risk functional over the restricted gain set and to use this uniformity to obtain convergence of the corresponding minimizers. Under a uniform exponential integrability assumption, the restricted entropic hedging gains converge in \(L^2\) to the unique restricted quadratic hedge as the risk-aversion parameter tends to zero. The theoretical result is illustrated numerically in a regularized discrete-time stochastic volatility model of Heston type.
\end{abstract}

\section{Introduction}

In incomplete financial markets, contingent claims cannot in general be replicated exactly. Hedging must therefore be formulated as an optimization problem over admissible trading strategies or terminal gains. A classical approach is quadratic hedging, where the objective is to minimize mean-squared hedging error. In discrete time this leads to an $L^2$-projection problem and to the variance-optimal hedging framework described, for instance, in~\cite[Section~10.3]{follmer2011stochastic} and~\cite{schweizer1995variance}.

A second approach is based on utility maximization and convex risk measures. In particular, entropic hedging is induced by exponential utility and leads to the entropic risk functional, which is a nonlinear criterion involving exponential moments of the terminal hedged payoff. The relation between exponential utility and mean-variance criteria is classical in the small-risk-aversion regime; see, e.g.,~\cite[Section~3]{pratt1978risk}. However, pointwise small-$\gamma$ expansions for fixed terminal payoffs are not sufficient for comparing hedging minimizers, since the entropic minimizer itself depends on $\gamma$.

The setting considered here is a discrete-time incomplete market in which admissible trading strategies are restricted to bounded finite-dimensional parametric classes. This restriction is natural in numerical hedging, where strategies are represented by finitely many basis coefficients and optimized over a bounded parameter domain, as in simulation-based methods in computational finance~\cite[Chapter~8]{glasserman2004monte} and in regression-based approaches such as~\cite{longstaff2001valuing}. The restriction is also mathematically useful as the induced admissible gain set is compact and convex in $L^2$.

The purpose of this paper is to compare restricted quadratic and restricted entropic hedging on this common admissible gain set. Under a uniform exponential integrability assumption, a uniform small-risk-aversion expansion of the entropic objective is proved over the whole restricted gain set. Uniformity allows the asymptotic expansion to be passed through the minimization problem. As a consequence, restricted entropic hedging gains converge in $L^2$ to the restricted quadratic hedge as $\gamma\downarrow0$.

The result is then illustrated numerically in a regularized discrete-time stochastic volatility model of Heston type~\citep{heston1993closed}. The model is modified by truncation and compensation so that discounted prices remain martingales and the exponential integrability assumption used in the analysis is satisfied.

\section{Market Model and Admissible Gains}

Fix a finite time grid
\[
0=t_0<t_1<\cdots<t_N=T
\]
and a filtered probability space
\[
(\Omega,\mathcal F,\mathbb Q,\mathbb F=(\mathcal F_k)_{k=0}^N).
\]
Unless otherwise stated, all expectations are taken under $\mathbb Q$. Let
\[
L^2:=L^2(\mathcal F_N)
\]
be the Hilbert space of square-integrable terminal random variables, with inner product
\[
\langle X,Y\rangle=\mathbb E[XY],
\qquad
\|X\|_2=\bigl(\mathbb E[X^2]\bigr)^{1/2}.
\]

Let $S=(S_k)_{k=0}^N$ denote the discounted price process of $d$ traded assets. The following standing assumption is used throughout.

\begin{assumption}\label{ass:price_martingale}
The discounted price process $S$ is adapted and is a square-integrable $\mathbb Q$-martingale. Thus, for $k=0,\dots,N-1$,
\[
\mathbb E[\Delta S_{k+1}\mid \mathcal F_k]=0,
\qquad
\Delta S_{k+1}:=S_{k+1}-S_k,
\]
and
\[
\mathbb E\left[\sum_{k=0}^{N-1}\|\Delta S_{k+1}\|^2\right]<\infty.
\]
\end{assumption}

A trading strategy is a predictable process $\delta=(\delta_k)_{k=0}^{N-1}$, where each $\delta_k$ is $\mathcal F_k$-measurable. Strategies are assumed to be self-financing. 
The corresponding discounted terminal trading gain is the discrete stochastic integral
\[
(\delta\cdot S)_N
=
\sum_{k=0}^{N-1}\delta_k^\top \Delta S_{k+1}.
\]
Let
\[
\Theta
:=
\left\{
\delta\text{ predictable}:
\mathbb E\left[\sum_{k=0}^{N-1}(\delta_k^\top \Delta S_{k+1})^2\right]<\infty
\right\}
\]
be the class of square-integrable predictable strategies. Define the stochastic integral operator
\[
\mathcal I:\Theta\to L^2(\mathcal F_N),
\qquad
\mathcal I(\delta):=(\delta\cdot S)_N,
\]
and the associated gain space
\[
G:=\mathcal I(\Theta)\subset L^2(\mathcal F_N).
\]

\begin{lemma}\label{lem:zero_mean}
Every admissible self-financing gain has zero mean. I.e., 
\[
\mathbb E[g] = 0,
\]
for all $g \in G$. 
\end{lemma}

\begin{proof}
Since $\delta_k$ is $\mathcal F_k$-measurable and $S$ is a $\mathbb Q$-martingale,
\[
\mathbb E[\delta_k^\top\Delta S_{k+1}]
=
\mathbb E\left[
\delta_k^\top\mathbb E[\Delta S_{k+1}\mid\mathcal F_k]
\right]
=0,
\]
and summing over $k$ yields the claim. 
\end{proof}

For a fixed contingent claim $H\in L^2(\mathcal F_N)$, the terminal hedged payoff generated by a gain $g$ is denoted
\[
X_g:=g-H.
\]
The market is incomplete when not every square-integrable claim can be represented as initial capital plus an attainable gain. In that case exact replication is generally impossible, and hedging is formulated as an optimization problem over admissible gains.

\section{Finite-Dimensional Strategy Restrictions}

Fix $n\in\mathbb N$. Let $V_n\subset \Theta$ be a finite-dimensional linear subspace with $m_n:=\dim(V_n)$, and fix a basis
\[
\{\varphi^1,\dots,\varphi^{m_n}\}\subset \Theta.
\]
Every strategy in $V_n$ can then be written uniquely as
\[
\delta^\theta
=
\sum_{j=1}^{m_n}\theta_j\varphi^j,
\qquad
\theta=(\theta_1,\dots,\theta_{m_n})\in\mathbb R^{m_n}.
\]

\begin{assumption}\label{ass:bounded_basis}
There exists $M_n>0$ such that
\[
\|\varphi_k^j\|_\infty\le M_n
\qquad
\text{for all }j=1,\dots,m_n,\quad k=0,\dots,N-1.
\]
\end{assumption}

For $q>0$, define the coefficient-bounded admissible strategy class
\[
V_n^A
:=
\{\delta^\theta\in V_n:\theta\in[-q,q]^{m_n}\}
\subset\Theta.
\]
The corresponding restricted admissible gain set is
\[
G_n^A
:=
\mathcal I(V_n^A)
=
\{\mathcal I(\delta^\theta):\theta\in[-q,q]^{m_n}\}
\subset L^2(\mathcal F_N).
\]

\begin{lemma}\label{lem:strategy_bound}
Under Assumption~\ref{ass:bounded_basis}, every $\delta\in V_n^A$ satisfies
\[
\|\delta_k\|_\infty\le m_n q M_n,
\qquad
k=0,\dots,N-1.
\]
\end{lemma}

\begin{proof}
If $\delta=\sum_{j=1}^{m_n}\theta_j\varphi^j$ with $\theta\in[-q,q]^{m_n}$, then
\[
\|\delta_k\|_\infty
\le
\sum_{j=1}^{m_n}|\theta_j|\,\|\varphi_k^j\|_\infty
\le
\sum_{j=1}^{m_n}qM_n
=
m_nqM_n.
\]
\end{proof}

\begin{lemma}\label{lem:GnA_properties}
The set $G_n^A$ is nonempty, compact, and convex in $L^2(\mathcal F_N)$. Moreover,
\[
\mathbb E[g]=0
\qquad
\text{for all }g\in G_n^A.
\]
\end{lemma}

\begin{proof}
The coefficient set $[-q,q]^{m_n}$ is compact and convex in $\mathbb R^{m_n}$. The map
\[
\theta\mapsto \mathcal I(\delta^\theta)
\]
from $\mathbb R^{m_n}$ to $L^2(\mathcal F_N)$ is linear. Since the domain is finite-dimensional, this map is bounded~\cite[Proposition~11.2]{brezis2011functional}, therefore continuous. Therefore $G_n^A$ is compact in $L^2$ as the continuous image of a compact set, and convex as the linear image of a convex set. It is nonempty since $0\in[-q,q]^{m_n}$ and $\mathcal I(\delta^0)=0$. Finally, each $g\in G_n^A$ is a self-financing martingale gain, hence has zero mean by Lemma~\ref{lem:zero_mean}.
\end{proof}

The compactness of $G_n^A$ is the main structural consequence of the bounded finite-dimensional restriction. It gives existence for the restricted entropic problem and compactness control for the small-$\gamma$ limit. Since $[-q,q]^{m_n}$ is not a linear subspace, $G_n^A$ is in general not a linear subspace of $L^2(\mathcal F_N)$.

\section{Restricted Quadratic and Entropic Hedging}

Quadratic hedging minimizes the mean-squared hedging error. In the unrestricted closed linear case, this is an orthogonal projection problem in $L^2$; see, e.g.,~\cite[Section~10.3]{follmer2011stochastic}. Under the bounded coefficient restriction, the feasible set $G_n^A$ is compact and convex but not necessarily linear. The restricted quadratic hedge is therefore the metric projection of the centered claim $H-\mathbb E[H]$ onto $G_n^A$.

Define
\[
g^{Q,n}
\in
\arg\min_{g\in G_n^A}
\left\|\mathbb E[H]+g-H\right\|_2^2.
\]

\begin{proposition}\label{prop:restricted_quadratic}
The restricted quadratic problem admits a unique minimizer $g^{Q,n}\in G_n^A$.
\end{proposition}

\begin{proof}
By Lemma~\ref{lem:GnA_properties}, $G_n^A$ is nonempty, compact, and convex in the Hilbert space $L^2(\mathcal F_N)$. In particular, it is closed and convex. The projection theorem for closed convex subsets of Hilbert spaces gives existence and uniqueness of the metric projection~\cite[Theorem~5.2]{brezis2011functional}.
\end{proof}

For $\gamma>0$, define the entropic risk functional
\[
\rho_\gamma(X)
:=
\frac{1}{\gamma}\log\mathbb E[e^{-\gamma X}],
\]
whenever the expectation is finite. This functional is the negative certainty equivalent associated with exponential utility; see, e.g.,~\cite[Section~2.3]{follmer2011stochastic}. The restricted entropic hedging problem is
\[
g_\gamma^{E,n}
\in
\arg\min_{g\in G_n^A}\rho_\gamma(X_g).
\]

Since $\rho_\gamma$ depends on exponential moments, square integrability alone is not enough. The following assumption supplies the tail control needed for finiteness, continuity, and the uniform small-$\gamma$ expansion.

\begin{assumption}\label{ass:uniform_exp_integrability}
There exists $\eta>0$ such that
\[
\sup_{g\in G_n^A}
\mathbb E[e^{\eta |X_g|}]
<\infty.
\]
\end{assumption}

\begin{proposition}\label{prop:restricted_entropic}
Assume Assumption~\ref{ass:uniform_exp_integrability}. Then for every $\gamma\in(0,\eta)$, the restricted entropic hedging problem admits at least one minimizer $g_\gamma^{E,n}\in G_n^A$.
\end{proposition}

\begin{proof}
For $\gamma\in(0,\eta)$, Assumption~\ref{ass:uniform_exp_integrability} gives finiteness of $\rho_\gamma(X_g)$ for all $g\in G_n^A$. If $g_m\to g$ in $L^2$, then $X_{g_m}\to X_g$ in probability. The same exponential integrability assumption gives uniform integrability of $e^{-\gamma X_{g_m}}$, hence
\[
\mathbb E[e^{-\gamma X_{g_m}}]\to \mathbb E[e^{-\gamma X_g}].
\]
Thus $g\mapsto \rho_\gamma(X_g)$ is continuous on $G_n^A$. Since $G_n^A$ is compact, the extreme value theorem gives existence of a minimizer.
\end{proof}

No uniqueness of $g_\gamma^{E,n}$ is required for the convergence result. The statements are made at the level of terminal gains. A parameter vector producing a given terminal gain need not be unique if the map $\theta\mapsto \mathcal I(\delta^\theta)$ is not injective.

\section{Small-Risk-Aversion Expansion}

For a fixed terminal payoff, the entropic risk functional admits a classical mean-variance expansion as $\gamma\downarrow0$~\cite[Section~3]{pratt1978risk}. Here the terminal payoff varies with the admissible gain $g$. Since the entropic minimizer depends on $\gamma$, a pointwise expansion for each fixed $g$ is not enough. The expansion must hold uniformly over the full feasible set $G_n^A$.

\begin{theorem}\label{thm:uniform_small_gamma}
Assume Assumption~\ref{ass:uniform_exp_integrability}. Then, as $\gamma\downarrow0$,
\[
\sup_{g\in G_n^A}
\left|
\rho_\gamma(X_g)
+
\mathbb E[X_g]
-
\frac{\gamma}{2}\Var(X_g)
\right|
=o(\gamma).
\]
Moreover, there exist $\gamma_1\in(0,\eta)$ and $C>0$ such that for all $\gamma\in(0,\gamma_1)$,
\[
\sup_{g\in G_n^A}
\left|
\rho_\gamma(X_g)
+
\mathbb E[X_g]
-
\frac{\gamma}{2}\Var(X_g)
\right|
\le C\gamma^2.
\]
\end{theorem}

\begin{proof}
Let
\[
Y_g:=X_g-\mathbb E[X_g].
\]
Then $\mathbb E[Y_g]=0$, $\mathbb E[Y^2_g]=\operatorname{Var}(X_g)$ and
\[
\rho_\gamma(X_g)
=
-\mathbb E[X_g]
+
\frac{1}{\gamma}\log\mathbb E[e^{-\gamma Y_g}].
\]
It is therefore enough to expand the final logarithmic term uniformly in $g$.

Fix $\gamma_0\in(0,\eta)$ and let $\gamma\in(0,\gamma_0]$. Taylor's theorem gives
\[
e^{-\gamma Y_g}
=
1-\gamma Y_g+\frac{\gamma^2}{2}Y_g^2+R_g,
\]
with
\[
|R_g|
\le
\frac{\gamma^3}{6}|Y_g|^3e^{\gamma_0|Y_g|}.
\]
Since Assumption~\ref{ass:uniform_exp_integrability} implies uniform exponential integrability of $X_g$, it also implies uniform bounds on all polynomial moments of $X_g$. In particular, $\sup_{g\in G_n^A}|E[X_g]|<\infty$. Hence, we have $|Y_g|\le |X_g|+C$ uniformly in $g$. Choosing $\gamma_0<\eta$, this gives
\[
\sup_{g\in G_n^A}E[|Y_g|^3 e^{\gamma_0 |Y_g|}]<\infty,
\]
which justifies the uniform Taylor remainder estimate. 
Hence
\[
\sup_{g\in G_n^A}
\left|
\mathbb E[e^{-\gamma Y_g}]
-
1
-
\frac{\gamma^2}{2}\Var(X_g)
\right|
=
O(\gamma^3).
\]
Define
\[
u_g(\gamma):=\mathbb E[e^{-\gamma Y_g}]-1.
\]
The previous estimate gives
\[
\sup_{g\in G_n^A}|u_g(\gamma)|=O(\gamma^2).
\]
For sufficiently small $\gamma$, this bound is uniformly smaller than $1/2$, so the logarithm can be expanded uniformly:
\[
\log(1+u_g(\gamma))
=
u_g(\gamma)+O(u_g(\gamma)^2).
\]
Since $u_g(\gamma)=O(\gamma^2)$ uniformly, it follows that
\[
\log\mathbb E[e^{-\gamma Y_g}]
=
\frac{\gamma^2}{2}\Var(X_g)+O(\gamma^3),
\]
uniformly over $g\in G_n^A$. Dividing by $\gamma$ yields
\[
\rho_\gamma(X_g)
=
-\mathbb E[X_g]
+
\frac{\gamma}{2}\Var(X_g)
+
O(\gamma^2),
\]
uniformly over $G_n^A$, which proves the stated expansion.
\end{proof}

\section{Convergence of Restricted Entropic Hedges}

The uniform expansion above identifies the small-$\gamma$ limit of the entropic objective. To pass from convergence of objective functions to convergence of minimizers, compactness of $G_n^A$ and uniqueness of the restricted quadratic minimizer are used.

For $\gamma\in(0,\eta)$, define the rescaled entropic objective
\[
J_{\gamma,n}(g)
:=
\frac{2}{\gamma}
\left(\rho_\gamma(X_g)+\mathbb E[X_g]\right),
\qquad
g\in G_n^A.
\]
Since $\mathbb E[g]=0$ for all $g\in G_n^A$, we have
\[
\mathbb E[X_g]=\mathbb E[g-H]=-\mathbb E[H],
\]
which is constant on $G_n^A$. Therefore $g\mapsto\rho_\gamma(X_g)$ and $g\mapsto J_{\gamma,n}(g)$ have the same minimizers on $G_n^A$. 
Moreover,
\[
\Var(X_g)
=
\mathbb E\left[\left(g-H+\mathbb E[H]\right)^2\right]
=
\left\|\mathbb E[H]+g-H\right\|_2^2.
\]
Thus the limiting objective in Theorem~\ref{thm:uniform_small_gamma} is exactly the restricted quadratic objective.

\begin{lemma}\label{lem:uniform_objective_convergence}
Under Assumption~\ref{ass:uniform_exp_integrability},
\[
\sup_{g\in G_n^A}
\left|
J_{\gamma,n}(g)-\Var(X_g)
\right|
\to0
\qquad
\text{as }\gamma\downarrow0.
\]
\end{lemma}

\begin{proof}
This follows directly from Theorem~\ref{thm:uniform_small_gamma} after multiplying the uniform expansion by $2/\gamma$.
\end{proof}

\begin{theorem}\label{thm:minimizer_convergence}
Assume Assumption~\ref{ass:uniform_exp_integrability}. Let $g^{Q,n}$ be the restricted quadratic minimizer, and let $g_\gamma^{E,n}$ be any restricted entropic minimizer for $\gamma\in(0,\eta)$. Then
\[
\|g_\gamma^{E,n}-g^{Q,n}\|_2\to0
\qquad
\text{as }\gamma\downarrow0.
\]
\end{theorem}

\begin{proof}
Let $\gamma_k\downarrow0$. Since $G_n^A$ is compact in $L^2(\mathcal F_N)$, the sequence $(g_{\gamma_k}^{E,n})_k$ has a subsequence, again denoted $(g_{\gamma_k}^{E,n})_k$, such that
\[
g_{\gamma_k}^{E,n}\to \bar g
\qquad
\text{in }L^2(\mathcal F_N)
\]
for some $\bar g\in G_n^A$.

Since $g_{\gamma_k}^{E,n}$ minimizes $J_{\gamma_k,n}$ on $G_n^A$,
\[
J_{\gamma_k,n}(g_{\gamma_k}^{E,n})
\le
J_{\gamma_k,n}(g^{Q,n}).
\]
Using Lemma~\ref{lem:uniform_objective_convergence}, this implies
\[
\Var(X_{g_{\gamma_k}^{E,n}})
\le
\Var(X_{g^{Q,n}})+o(1).
\]
Passing to the limit along the subsequence and using continuity of variance under $L^2$-convergence gives
\[
\Var(X_{\bar g})
\le
\Var(X_{g^{Q,n}}).
\]
Hence $\bar g$ is a minimizer of the restricted quadratic objective. By Proposition~\ref{prop:restricted_quadratic}, this minimizer is unique, so $\bar g=g^{Q,n}$.

Every convergent subsequence of every sequence $\gamma_k\downarrow0$ therefore has the same limit $g^{Q,n}$. Since all minimizers lie in the compact set $G_n^A$, the whole family converges:
\[
\|g_\gamma^{E,n}-g^{Q,n}\|_2\to0.
\]
\end{proof}

The theorem gives convergence at the level of terminal gains. It does not imply convergence of the coefficient vectors $\theta_\gamma$ unless the map $\theta\mapsto \mathcal I(\delta^\theta)$ is injective or otherwise identifiable on the admissible class.

\section{Numerical Illustration}

The theoretical result is illustrated in a discrete-time stochastic volatility model of Heston type~\citep{heston1993closed}. The purpose of the experiment is to test the small-$\gamma$ behavior on a fixed admissible gain set $G_n^A$.

The numerical model uses one traded asset and one non-traded variance factor, which gives an incomplete-market setting. The variance dynamics are of Euler-Heston type, with correlated innovations as in standard simulation schemes for stochastic volatility models~\cite{lord2010comparison}. Since the entropic objective requires exponential integrability of terminal hedged payoffs, the log-price increments are truncated. A predictable compensator is then added so that the discounted price process remains a $\mathbb Q$-martingale.

The contingent claim is a European call option
\[
H=(S_T-K)^+.
\]
Under the truncated compensated dynamics, the log-price increments are bounded on the finite time grid. Together with the bounded coefficient restriction on strategies, this gives a uniform bound on all admissible gains:
\[
|g|
=
\left|
\sum_{k=0}^{N-1}\delta_k^\top\Delta S_{k+1}
\right|
\le C_g,
\qquad
g\in G_n^A.
\]
The call payoff is also bounded under the truncated dynamics. Hence $|X_g|=|g-H|$ is uniformly bounded over $G_n^A$, and Assumption~\ref{ass:uniform_exp_integrability} follows.

For each Monte Carlo seed, simulated paths are split into training and test sets. The restricted quadratic hedge and the restricted entropic hedges are computed on the training set for a sequence of decreasing risk-aversion parameters $\gamma$. The resulting terminal gains are then evaluated on the test set.

The main convergence metric is the normalized gain-space distance
\[
\frac{\|g_\gamma^{E,n}-g^{Q,n}\|_2}
{\|H-\widehat{\mathbb E}[H]\|_2},
\]
where $\widehat{\mathbb E}[H]$ is the empirical payoff mean. This metric corresponds directly to the $L^2$ convergence statement in Theorem~\ref{thm:minimizer_convergence}. 

\begin{figure}[htbp]
\centering
\IfFileExists{figures/plot_gamma_distance_reg.png}{%
\includegraphics[width=0.8\textwidth]{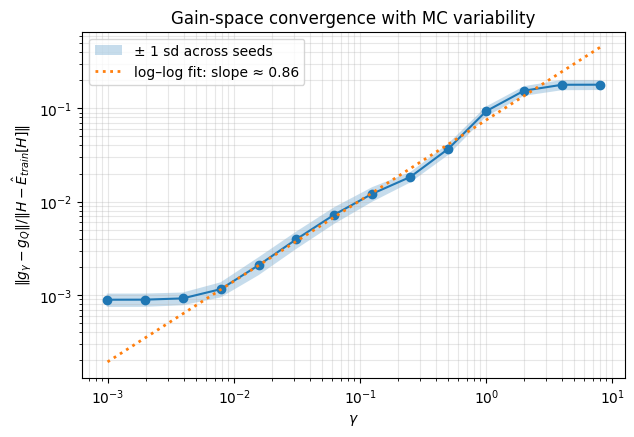}%
}{%
\fbox{\parbox{0.8\textwidth}{\centering Placeholder for \texttt{figures/plot\_gamma\_distance\_reg}}}%
}
\caption{Normalized distance between restricted entropic and quadratic gains as a function of $\gamma$. The dotted line shows a log-log regression over an intermediate range of $\gamma$, and the shaded region indicates sampling variability across Monte Carlo seeds.}
\label{fig:gamma-distance}
\end{figure}

Figure~\ref{fig:gamma-distance} shows that the gain-space distance decreases as $\gamma$ becomes small. In the intermediate small-$\gamma$ range, the decay is approximately linear on a log-log scale. For very small $\gamma$, the curve reaches a plateau caused by Monte Carlo sampling error and numerical precision. The log-log fit gives an empirical exponent of approximately $0.86$ over the fitted range. This exponent is only a numerical description of the experiment since Theorem~\ref{thm:minimizer_convergence} proves convergence but does not provide a rate.

The convergence statement is formulated at the level of terminal gains. It does not, in general, imply convergence of the corresponding coefficient vectors, since the map $\theta\mapsto \mathcal I(\delta^\theta)$ need not be injective. Parameter-space convergence would require additional identifiability assumptions on the chosen parametrization.

Overall, the numerical results are consistent with the theoretical prediction that restricted entropic hedging gains converge to the restricted quadratic hedge in the small-risk-aversion limit.

\section{Conclusion}

This paper studied quadratic and entropic hedging in a discrete-time incomplete market under bounded finite-dimensional strategy restrictions. The restriction induces an admissible gain set $G_n^A$ that is compact and convex in $L^2$. This gives a common well-posed setting for the restricted quadratic problem and the restricted entropic problem.

The main analytical step was a uniform small-risk-aversion expansion of the entropic risk functional over $G_n^A$. The uniformity is essential because the entropic minimizer depends on the risk-aversion parameter. After rescaling and removing constants independent of the admissible gain, the entropic objective converges uniformly to the quadratic objective on $G_n^A$.

Combining this uniform objective convergence with compactness of the admissible gain set and uniqueness of the restricted quadratic minimizer gives convergence of restricted entropic hedging gains to the restricted quadratic hedge in $L^2$ as $\gamma\downarrow0$.

The numerical experiment in a truncated compensated Heston-type model illustrates the same behavior. The gain-space distance between the restricted entropic hedge and the restricted quadratic hedge decreases as the risk-aversion parameter becomes small, until the difference is dominated by Monte Carlo and numerical error.

The result is limited to a fixed finite-dimensional bounded strategy class and to convergence at the level of terminal gains. It does not by itself imply convergence of strategy parameters, nor does it address approximation as the dimension of the strategy class increases. These are separate questions from the small-risk-aversion limit considered here.

\section*{Acknowledgements}

This paper is based on parts of the author's master's
thesis~\citep{hanetho2026quadratic}, completed at the
University of Oslo in 2026. The author thanks Giulia Di Nunno
for her supervision and helpful feedback.

\bibliographystyle{apalike}
\bibliography{sample2}

@book{follmer2011stochastic,
  title={Stochastic finance: an introduction in discrete time},
  author={F{\"o}llmer, Hans and Schied, Alexander},
  year={2011},
  publisher={Walter de Gruyter}
}

@article{lord2010comparison,
  title={A comparison of biased simulation schemes for stochastic volatility models},
  author={Lord, Roger and Koekkoek, Remmert and Dijk, Dick Van},
  journal={Quantitative Finance},
  volume={10},
  number={2},
  pages={177--194},
  year={2010},
  publisher={Taylor \& Francis}
}

@article{schweizer1995variance,
  title={Variance-optimal hedging in discrete time},
  author={Schweizer, Martin},
  journal={Mathematics of Operations Research},
  volume={20},
  number={1},
  pages={1--32},
  year={1995},
  publisher={INFORMS}
}

@book{brezis2011functional,
  title={Functional analysis, Sobolev spaces and partial differential equations},
  author={Brezis, Haim},
  volume={2},
  number={3},
  year={2011},
  publisher={Springer}
}

@incollection{pratt1978risk,
  title={Risk aversion in the small and in the large},
  author={Pratt, John W},
  booktitle={Uncertainty in economics},
  pages={59--79},
  year={1978},
  publisher={Elsevier}
}

@article{heston1993closed,
  title={A closed-form solution for options with stochastic volatility with applications to bond and currency options},
  author={Heston, Steven L},
  journal={The review of financial studies},
  volume={6},
  number={2},
  pages={327--343},
  year={1993},
  publisher={Oxford University Press}
}

@book{glasserman2004monte,
  title={Monte Carlo methods in financial engineering},
  author={Glasserman, Paul},
  volume={53},
  year={2004},
  publisher={Springer}
}

@article{longstaff2001valuing,
  title={Valuing American options by simulation: A simple least-squares approach},
  author={Longstaff, Francis A and Schwartz, Eduardo S},
  journal={The review of financial studies},
  volume={14},
  number={1},
  pages={113--147},
  year={2001},
  publisher={Oxford University Press}
}

@mastersthesis{hanetho2026quadratic,
  author = {Hanetho, Jonas Rotschi},
  title  = {Quadratic and Entropic Hedging in Incomplete Markets:
            Finite-Dimensional Strategy Restrictions and Asymptotic Analysis},
  school = {University of Oslo},
  year   = {2026},
  type   = {Master's thesis},
  url    = {https://hdl.handle.net/11250/5542261}
}

\end{document}